\documentclass[12pt,reqno,twoside]{amsart}
\usepackage{amsmath,amssymb,amsfonts,amsthm,indentfirst,setspace}
\usepackage{xcolor}\usepackage{ulem}
\usepackage{graphicx}
\newtheorem{theorem}{Theorem}[section]

\newtheorem{corollary}{Corollary}[section]

\newcommand{\be}{\begin{equation}}
\newcommand{\ee}{\end{equation}}
\newcommand{\bea}{\begin{eqnarray}}
\newcommand{\eea}{\end{eqnarray}}
\newcommand{\eeas}{\end{eqnarray*}}
\newcommand{\beas}{\begin{eqnarray*}}

\newcommand{\G}{\mathcal{G}}

\begin{document}

\title[How a conformally flat $(GR)_4$ impacts Gauss-Bonnet gravity?]{HOW A CONFORMALLY FLAT $(GR)_4$ IMPACTS GAUSS-BONNET GRAVITY?}
\author{Avik De \and Tee-How Loo \and Raja Solanki \and P.K. Sahoo}
\address{A. De\\
Department of Mathematical and Actuarial Sciences\\
Universiti Tunku Abdul Rahman\\
Jalan Sungai Long\\
43000 Cheras\\
Malaysia}
\email{de.math@gmail.com}
\address{T. H. Loo\\
Institute of Mathematical Sciences\\
Universiti Malaya\\
50603 Kuala Lumpur\\
Malaysia}
\email{looth@um.edu.my}
\address{Raja Solanki\\
Department of Mathematics, Birla Institute of Technology and Science-Pilani, Hyderabad Campus, Hyderabad 500078\\ 
India}
\email{rajasolanki8268@gmail.com}
\address{P.K.Sahoo\\
Department of Mathematics, Birla Institute of Technology and Science-Pilani, Hyderabad Campus, Hyderabad 500078\\ 
India}
\email{pksahoo@hyderabad.bits-pilani.ac.in}
\date{}
\maketitle

\begin{abstract}

First and foremost, we show that a $4$-dimensional conformally flat generalized Ricci recurrent spacetime $(GR)_4$ is an Einstein manifold. We examine such a spacetime as a solution of $f(R,\mathcal{G})$-gravity theory and it is shown that the additional terms from the modification of the gravitational sector can be expressed as a perfect fluid. Several energy conditions are investigated with $f(R,\mathcal{G})= R+\sqrt{\mathcal{G}}$ and $f(R,\mathcal{G})= R^2+\mathcal{G}ln\mathcal{G}$. For both the models, weak, null and dominant energy conditions are satisfied while strong energy condition is violated, which is a good agreement with the recent observational studies which reveals that the current universe is in accelerating phase.   
\end{abstract}
\maketitle
\textbf{Key words: } Modified Gauss-Bonnet gravity, Generalised Ricci recurrent, conformally flat, Energy conditions, $f(R,G)$ gravity

\section{\textbf{Introduction}}
Patterson considered a Riemannian manifold $M^n, \, (n>2)$ whose Ricci curvature tensor $R_{ij}$ is not zero  everywhere and also complies with $\nabla_{i}R_{jl}= A_iR_{jl}$, with 1-form $A_i (\neq 0)$, and called it a Ricci recurrent manifolds $R_n$ \cite{patterson}. The initial idea might be from purely mathematical curiosity fuelled by the complete classification of locally symmetric ($\nabla_lR_{ijkl}=0$) spaces, but by then interaction of Ricci calculus and theory of gravity were already increasing. After around forty years, in 1995, the idea was generalized by De et al. \cite{de} when they presented the notion of a generalized Ricci recurrent manifold as a non-flat Riemannian manifold of dimension $n$ whose Ricci curvature tensor fulfils:
\be \nabla _{i}R_{jl}= A_iR_{jl}+B_ig_{jl},\label{gr}\ee
where $A_i$ and $B_i$ are two non-zero 1-forms. It reduces to $R_n$ for a null $B_i$. A Lorentzian manifold of signature $(-,+,+,+)$ is a generalized Ricci recurrent spacetime $(GR)_4$ if its Ricci curvature tensor agrees with (\ref{gr}). 

A $(GR)_4$ spacetime was shown to be a generalized Robertson Walker spacetime with Einstein fiber if the Ricci curvature tensor is of Codazzi type \cite{pinaki}, also  for a constant Ricci scalar, it was shown that the stress-energy tensor is of semisymmetric type. Recently a $(GR)_4$ spacetime was studied in general relativistic cosmology in \cite{avik}. Under certain conditions a conformally flat $(GR)_4$ was shown to be a perfect fluid. Continuing to this study, in \cite{avik-2}, the present authors studied a conformally flat 
$(GR)_4$ with constant $R$ in $F(R)$ theory. The presently accepted homogeneous and isotropic model of our universe, the Robertson-Walker type spacetime is $(GR)_4$ if and only if the Ricci curavture tensor is parallel ($\nabla_kR_{ij}=0$). It is also shown that the equation of state (EoS) $\omega = -1$. 

With the breakthrough observation about the accelerated universe, the limitation of the standard theory of gravity governed by Einstein's field equations (EFE), $R_{ij}-\frac{R}{2}g_{ij}=\kappa T^{(m)}_{ij}$, is quite visible unless the dark energy is solely accounted for this accelerating expansion.  The simplest candidate for dark energy is the cosmological constant $\Lambda$ i.e. the fluid responsible for high negative pressure. The model with cosmological constant $\Lambda$ is characterized by equation of state $\omega_\Lambda=-1$ and known as $\Lambda$CDM model \cite{SMC,W}. Even though $\Lambda$CDM model agrees with observational data, it has some drawbacks like the fine-tuning problem. The fine-tuning problem refers to the instability between the theoretically predicted value and the observed value of the cosmological constant \cite{Cop} . This motivated some researchers to modify the gravitational sector of the EFE and to explain the negative pressure as a curvature effect \cite{nojiri,nojiri1,nojiri2}. Replacing the Ricci scalar $R$ with a function $f(R,\mathcal{G})$ of $R$ and of the Gauss-Bonnet topological invariant 
$$\mathcal{G}=R_{ijkl}R^{ijkl}-4R_{kl}R^{kl}+R^2.$$
in the action term
\[S=\frac{1}{2\kappa}\int d^4x\sqrt{-g}f(R,\mathcal{G})+S^{\text{matter}},\] 
where the $S^{\text{matter}}$ is the action term of the standard matter fields, the least action principle gives the field equations of gravity 
\begin{align}\label{gb0}
0=&\kappa T^{(m)}_{ik}+\frac12fg_{ik}	
			+f_\G\Big(-2R_{ik}R+4R_{mink}R^{mn}-2 R_{impq}R_{k}{}^{mpq}+4R^m_iR_{mk}\Big) \notag\\
&+\Big(2R\nabla_k\nabla_i-2g_{ik}R\Box-4R^m_i\nabla_m\nabla_k-4R^m_k \nabla_m\nabla_i 
									+4R_{ik} \Box	\Big)f_\G  \notag\\
			&+\Big(4g_{ik}R_{mn} \nabla^n\nabla^m 		 -4R_{nimk}\nabla^n\nabla^m 
						\Big)f_\G   
					-(R_{ik}-\nabla_i\nabla_k+g_{ik}\Box)f_R,
\end{align}
where the term $T^{(m)}_{ik}$ produces from $S^{\text{matter}}$ and
\[
f_R=\frac{\partial f}{\partial R}; \qquad f_\G=\frac{\partial f}{\partial \G}.
\]
In most cases the stress-energy tensor is considered to be of perfect fluid form 
\be T^{(m)}_{ij}=(p^{(m)}+\rho^{(m)})u_iu_j+p^{(m)}g_{ij}.\ee
$p^{(m)}$ and $\rho^{(m)}$ respectively denote the isotropic pressure and energy density of the perfect fluid. We further assume that the four velocity vector of the perfect fluid type stress-energy tensor coincides with $A^i$.

In particular, when the spacetime is conformally flat, 
The equation (\ref{gb0}) can be rewritten as 
\be 
R_{ik}-\frac1{2}Rg_{ik}=\kappa (T^{(m)}_{ik}+T^{\text{curv}}_{ik})=\kappa T^{\text{eff}}_{ik},\label{gb}
\ee
where the term $T^{\text{curv}}_{ik}$ results from the geometry of the spacetime and is given by:
\begin{align*}
\kappa T^{\text{curv}}_{ik}
=&(\nabla_i\nabla_k-g_{ik}\Box)f_R+2R(\nabla_i\nabla_k-g_{ik}\Box)f_\mathcal{G}
			-4(R^m_i\nabla_m\nabla_k+R^m_k\nabla_m\nabla_i)f_\mathcal{G} \notag\\
 &+4(R_{ik}\Box+g_{ik}R_{mn}\nabla^n\nabla^m-R_{nimk}\nabla^n\nabla^m)f_\mathcal{G}\notag\\
 &-\frac{1}{2}g_{ik}(Rf_R+\G f_\G-f)+(1-f_R)\Big(R_{ik}-\frac{R}{2}g_{ik}\Big).
\end{align*}

To explain the inflationary epoch of the universe, the $f(R)$ theory of gravity gained much attention.  Another modification of GR uses Lovelock invariant, such as the Gauss Bonnet scalar $\mathcal{G}$. It can successfully describe both the inflationary and the dark energy epoch \cite{BL,Noj-2,Noj-3}. It is believed that the mechanism of inflation can be improved by taking the combination of the two scalars $R$ and $\mathcal{G}$. In cosmology, the framework of Gauss-Bonnet gravity adds higher order terms in the Hubble parameter which have main consequences in the early epoch of the universe \cite{observe}. The  $f(R,\mathcal{G})$ theory of gravity has gained much attention of cosmologists \cite{DF,MM,SC,pheno}. Azhar et al. have investigated the impact of $f(R,\mathcal{G})$ gravity on gravitational baryogenesis and observational bounds \cite{Az}. Alvaro et al. have analyzed the stability of the cosmological solutions in the modified $f(R,\mathcal{G})$ gravity \cite{Al}. Apart from phenomenological approaches, first principles like energy conditions, causal structure, and the classification of singularities can be considered to restrict the possible forms of $f(R,\mathcal{G})$.\\
The different energy conditions play a crucial role to discuss many important issues in cosmology. It is very helpful in the proof of some theorems such as the law of black hole thermodynamics or no hair theorem \cite{SW}. These different energy conditions can be obtained with the help of the very popular Raychaudhuri equation. The energy conditions were originally obtained in the standard theory of gravity. By introducing an effective energy density and effective pressure one can obtain these conditions in $f(R,\mathcal{G})$ gravity. Plenty of authors have studied the energy conditions in different modified gravity theories, see the references \cite{SM,SA,JS,MS}. 

The present paper is organized as follows: after introduction, in Section 2 we show that a conformally flat $(GR)_4$ spacetime is an Einstein manifold. We study the Gauss-Bonnet gravity theory in such a spacetime and derive the effective pressure and energy density. In Section 3, we discuss some energy conditions by taking two different $f(R,\mathcal{G})$ models. Finally, in the last section we discuss our results.    


\section{\textbf{Curvature tensors in conformally flat $(GR)_4$}}
In this section we thoroughly investigate the curvature properties of conformally flat $(GR)_4$. We prove the following result without any additional assumption on the spacetime or the associated vectors which is a common practice while studying this topic.

\begin{theorem}\label{res1}
A conformally flat $(GR)_4$ is an Einstein manifold, in particular, 
$R_{kl}=-\varepsilon\lambda g_{kl}$
where $\lambda=A^iB_i$ and $A^iA_i=\varepsilon=\pm1$, provided $A^i$ is a unit vector.
Furthermore, it is either a de Sitter space or an anti-de Sitter space.
\end{theorem} 
\begin{proof}
Since the spacetime is generalized Ricci recurrent, the Ricci curvature tensor $R_{kl}$ must satisfy (\ref{gr})
with one additional assumption that $A^i$ is a unit vector with $A^iA_i=\varepsilon=\pm1$ . 
Contracting $j$ and $l$ in (\ref{gr}) we obtain
\be \nabla_iR=RA_i+4B_i.\label{a1}\ee
Contracting  $i$ and $l$ in (\ref{gr}) we obtain
\be \frac{1}{2}\nabla_jR=A_lR^l_j+B_j.\label{a2}\ee
It follows from (\ref{a1}) and (\ref{a2}) that 
\be 2A_lR^l_i=RA_i+2B_i.\ee
On the other hand, from the conformal flatness of the spacetime we have
\[
\nabla_i R_{kl}-\nabla_k R_{il}=\frac{g_{kl}\nabla_iR-g_{il}\nabla_kR}{6}.
\]
Using (\ref{gr}) and (\ref{a1}) in the above equation we calculate
\be 
6A_iR_{kl}-6A_kR_{il}=RA_ig_{kl}-RA_kg_{il}-2B_ig_{kl}+2B_kg_{il},
\label{a4}\ee
Transvecting (\ref{a4}) with $A^i$ we obtain
\bea
6\varepsilon  R_{kl}=\varepsilon Rg_{kl}-2A^iB_ig_{kl}+2RA_kA_l+2A_lB_k+6A_kB_l.
\label{a6}\eea
Swapping $k$ and $l$ in (\ref{a6}) we obtain
\[
6{\varepsilon}  R_{lk}={\varepsilon}  Rg_{lk}-2A^iB_ig_{lk}+2RA_lA_k+2A_kB_k+6A_lB_k.
\]
These two equations imply
\be B_i={\varepsilon}\lambda A_i, \quad (\lambda=  A^lB_l),\label{a8}\ee
Armed with this significant result, using (\ref{a6}) we can conclude that the Ricci curvature tensor and the Ricci scalar satisfy
\be 6R_{kl}=(R-2{\varepsilon}  \lambda)g_{kl}
			+(2R+8{\varepsilon}\lambda){\varepsilon}  A_kA_l, \quad	
			2A_lR^l_i=(R+2{\varepsilon}  \lambda)A_i										\label{a9}\ee
\be \nabla_iR=(R+4{\varepsilon}  \lambda)A_i, \label{a10}\ee
Covariantly differentiating (\ref{a9}) with respect to $i$ we obtain
\begin{align*}
6\nabla_iR_{kl}
=& (\nabla_iR-2{\varepsilon}\nabla_i\lambda)g_{kl}
				+(2\nabla_iR+8{\varepsilon}\nabla_i\lambda){\varepsilon}  A_kA_l
				+(2R+8{\varepsilon} \lambda){\varepsilon}\nabla_i(A_kA_l)\\
=& \{(R+4{\varepsilon}\lambda)A_i-2{\varepsilon}\nabla_i\lambda\}g_{kl}
				+\{(2R+8{\varepsilon}\lambda)A_i
						+8{\varepsilon}\nabla_i\lambda\}{\varepsilon}A_kA_l	\\
 &+(2R+8{\varepsilon} \lambda){\varepsilon}\nabla_i(A_kA_l).
\end{align*} 
On the other hand,  (\ref{gr}), (\ref{a8}) and (\ref{a9}) give
\begin{align*}
6\nabla_iR_{kl}
=& (R+4{\varepsilon}\lambda)A_ig_{kl}
				+(2R+8{\varepsilon}\lambda){\varepsilon}A_iA_kA_l.	
\end{align*} 
These two equations give
\begin{align*}
 -{\varepsilon}g_{kl}\nabla_i\lambda +4A_kA_l\nabla_i\lambda 	
  +(R+4{\varepsilon} \lambda){\varepsilon}\nabla_i(A_kA_l)=0.
\end{align*}
Transvecting with $A^k$ and $A^l$ we obtain
\begin{align}\label{a11}
\nabla_i\lambda=0.
\end{align}
It follows from these two equations that 
\begin{align}\label{a12}
(R+4{\varepsilon} \lambda)\nabla_i(A_kA_l)=0.
\end{align}

We consider two cases: $R+4{\varepsilon} \lambda\neq0$ and $R+4{\varepsilon} \lambda=0$.
\begin{description}
\item[a] If $R+4{\varepsilon} \lambda\neq0$, we have $\nabla_i A_l=0$ and so $A_lR_i^l=0$.
It follows from (\ref{a9}) that 
\be \label{bbb} R=-2{\varepsilon}\lambda.
\ee
It follows from this equation and (\ref{a11}) that 
\[
\nabla_iR=-2{\varepsilon}\nabla_i\lambda=0.
\]
After applying this to (\ref{a10}) gives
$R+4{\varepsilon}\lambda=0$; a contradiction. 
Hence this case is impossible.
\item[b] Suppose that  
\be R+4{\varepsilon}\lambda=0.
\label{rscalar}
\ee
Therefore, using (\ref{a9}) we obtain
\be 
R_{kl}=-{\varepsilon}\lambda g_{kl},
\label{ricci}
\ee
meaning that it is an Einstein space. As a conformaly flat Einstein space is of constant sectional curvature.
We conclude that a conformally flat generalized Ricci recurrent spacetime is either a de Sitter space or an anti-de Sitter space.
\end{description}
\end{proof}

\begin{corollary}\label{res2}
In a conformally flat $(GR)_4$, the Gauss-Bonnet term can be expressed as $$\mathcal{G}=\frac{8}{3}\lambda^2,$$ provided the unit vector $A^i$ is not parallel.
\end{corollary}
\begin{proof}
In any $4$-dimensional conformally flat spacetime, $\mathcal{G}$ can be expressed as
\be\mathcal{G}=-2R_{kl}R^{kl} +\frac{2}{3}R^2\label{gconf},\ee
by (\ref{ricci}) which reduces to 
\[ \mathcal{G}=\frac{8}{3}\lambda^2.\] 
\end{proof}
\begin{theorem}
In a conformally flat $(GR)_4$ satisfying $f(R,\mathcal{G})$-gravity, if the four-velocity vector is identical with $A^i$, then the effective pressure $p^{\text{eff}}$ and the effective energy density $\rho^{\text{eff}}$ are given by 
$$p^{\text{eff}} = p^{(m)}+\frac{1}{\kappa}\left(\frac{f}{2}-\lambda f_R-\frac{4\lambda^2}{3}f_\mathcal{G}-\lambda\right),$$
$$\rho^{\text{eff}}=\rho^{(m)}-\frac{1}{\kappa}\left(\frac{f}{2}-\lambda f_R-\frac{4\lambda^2}{3}f_\mathcal{G}-\lambda\right).$$
\end{theorem}
\begin{proof}
Using the results of Theorem \ref{res1} and Corollary \ref{res2} in the Gauss-Bonnet gravity equations (\ref{gb}), after some calculations we obtain
\be \kappa T^{\text{eff}}_{ik}=\kappa T^{\text{matter}}_{ik}
+\left(\frac{f}{2}-\lambda f_R-\frac{4\lambda^2}{3}f_\mathcal{G}-\lambda\right)g_{ik	}.\ee
Hence the result.
\end{proof}
\section{\textbf{Constraints on $f(R,\mathcal{G})$-gravity in a $(GR)_4$}}

Several models of $f(R,\mathcal{G})$ gravity theories have been proposed in the literature. In this section, we consider two different models of $f(R,\mathcal{G})$ gravity theories to analyze our results. To examine different energy conditions for both the models we consider the vacuum case for simplicity \cite{K.A}, that is, $\rho^{(m)}= p^{(m)}=0$.
\\
\\
\textbf{Case:I} $f(R,\mathcal{G})= R+\sqrt{\mathcal{G}}$. R. Myrzakulov et al \cite{R.M} also considered this type of $f(R,\mathcal{G})$ model. In this case, the effective energy density and the effective pressure for a perfect fluid matter as follows: 

\begin{equation}
\rho^{eff}=\frac{1}{\kappa} \left[ 2\lambda+ \frac{2\lambda^2}{3\sqrt{\mathcal{G}}}-\left(\frac{ R+ \sqrt{\mathcal{G}}}{2} \right) \right]
\end{equation}
and
\begin{equation}
p^{eff}=\frac{1}{\kappa} \left[\left(\frac{ R+ \sqrt{\mathcal{G}}}{2} \right)-2\lambda - \frac{2\lambda^2}{3\sqrt{\mathcal{G}}} \right]
\end{equation}
\\
We know that the relationship between energy density and pressure given by $\frac{p}{\rho}$ is known as the equation of state parameter(EoS). In this case, the equation of state parameter turns out to be $\omega=-1$. Hence the universe is dominated by cosmological constant and the given model is consistent with well established $\Lambda$CDM theory. As weak energy condition(WEC) is the combination of positive density and null energy condition(NEC) and in the present setting NEC is zero, thus we examine the behavior of dominant energy condition(DEC), strong energy condition(SEC), and the density parameter.

\begin{figure}[h]
{\includegraphics[scale=0.55]{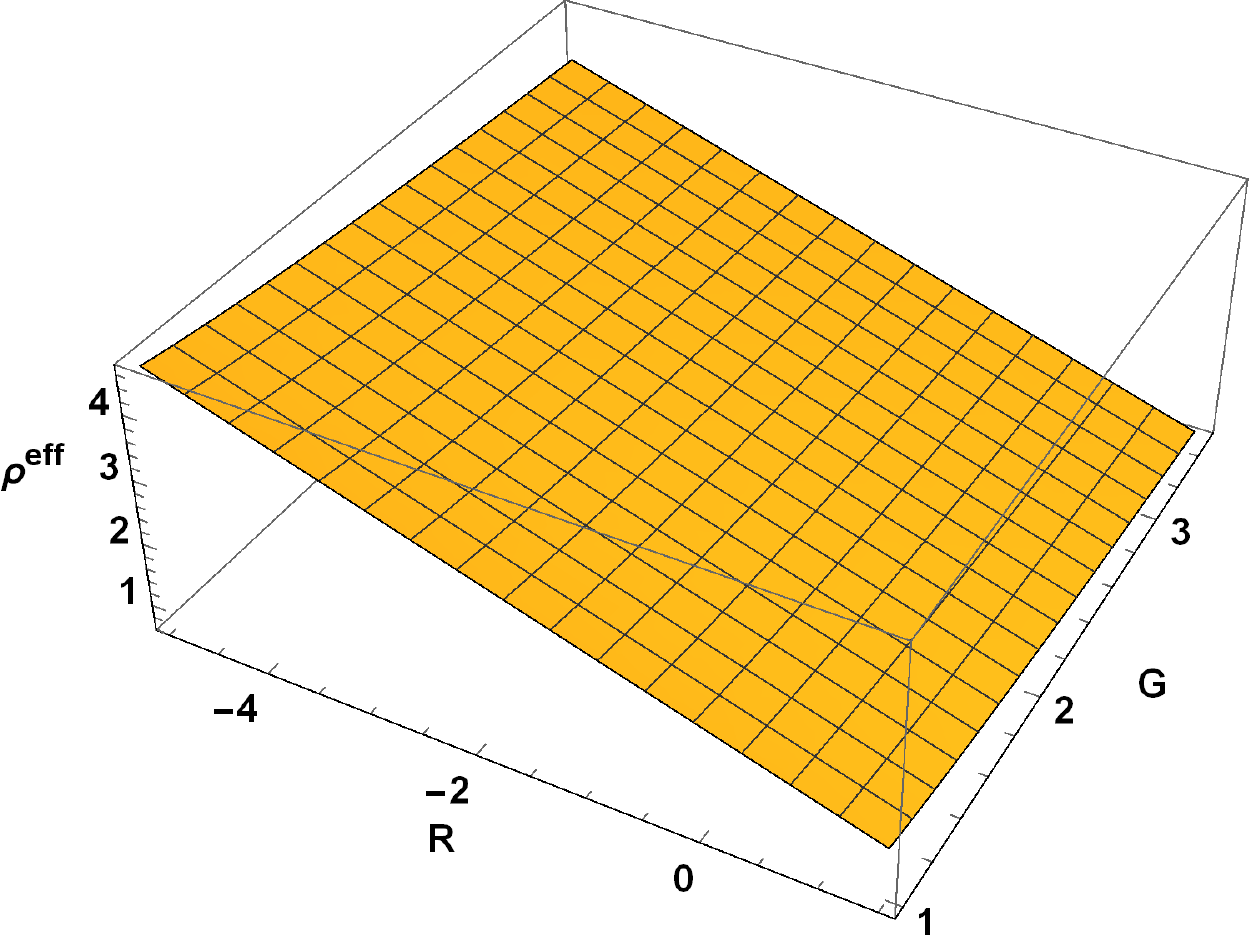}}
{\includegraphics[scale=0.55]{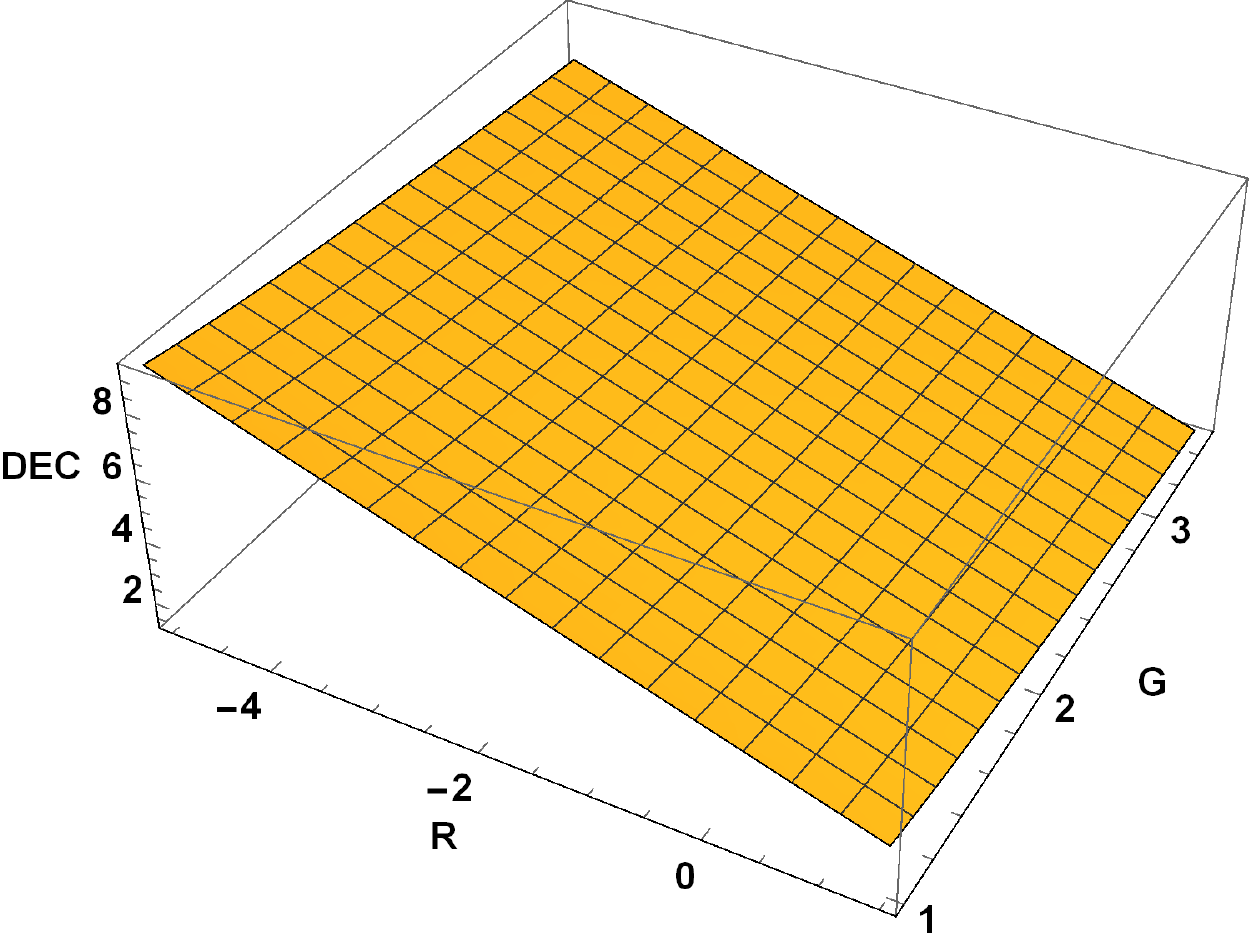}}
{\includegraphics[scale=0.55]{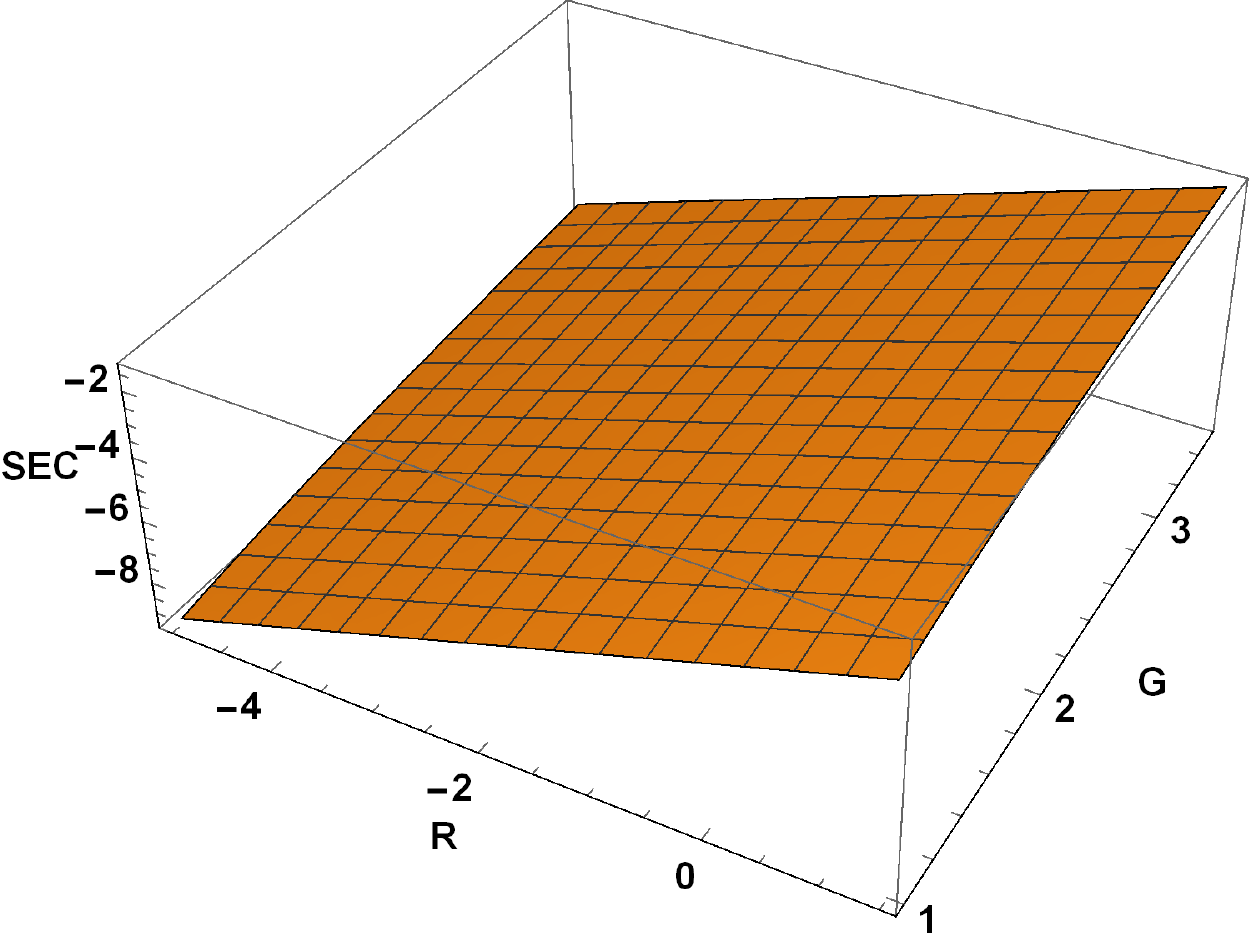}}
\caption{Energy conditions for $f(R,\mathcal{G})=R+\sqrt{\mathcal{G}}$ with $\lambda=1$, $-5\leq R \leq 1.5$ and $1\leq \mathcal{G} \leq 3.5 $}\label{fig1}
\end{figure}
\newpage
From Figure \ref{fig1} it is clear that the behavior of density parameter is positive, DEC also shows the positive behavior while the SEC violates it, which represents the accelerated expansion phase of the universe.  
\\
\\
\textbf{Case:II} $f(R,\mathcal{G})= R^2+\mathcal{G}ln\mathcal{G}$. Sebastian Bahamonde et al \cite{SB}  also considered this type of $f(R,\mathcal{G})$ model. In this case, the effective energy density and the effective pressure for a perfect fluid matter as follows: 
 
\begin{equation}
\rho^{eff}=\frac{1}{\kappa} \left[ \lambda(1+2R)+\frac{4\lambda^2}{3}(1+ln\mathcal{G})- \left( \frac{R^2+ \mathcal{G}ln\mathcal{G}}{2} \right) \right]
\end{equation}
and
\begin{equation}
p^{eff}=\frac{1}{\kappa} \left[ \left( \frac{R^2+ \mathcal{G}ln\mathcal{G}}{2} \right)-\lambda(1+2R)-\frac{4\lambda^2}{3}(1+ln\mathcal{G}) \right]
\end{equation}

\begin{figure}[h]
{\includegraphics[scale=0.55]{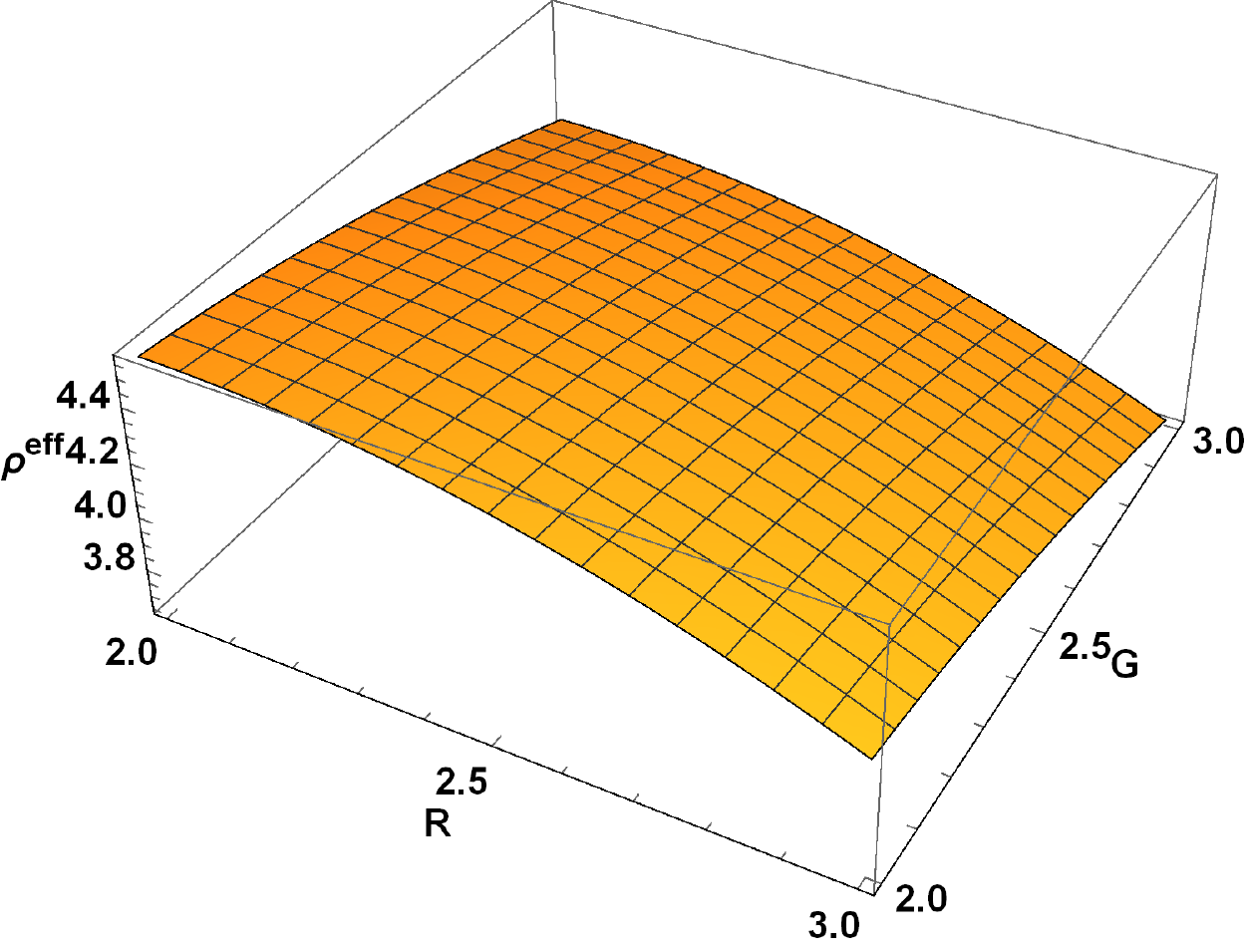}}
{\includegraphics[scale=0.55]{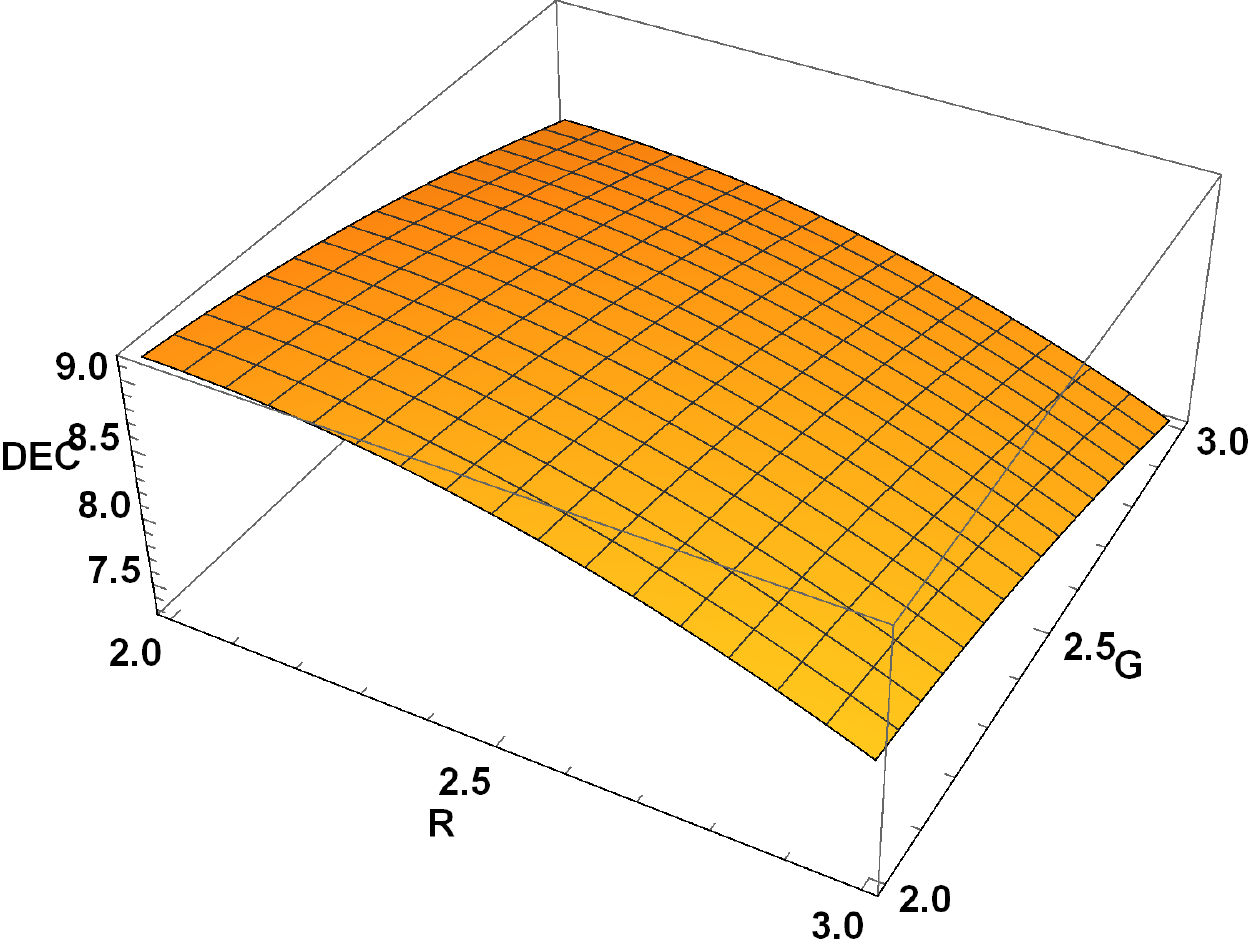}}
{\includegraphics[scale=0.55]{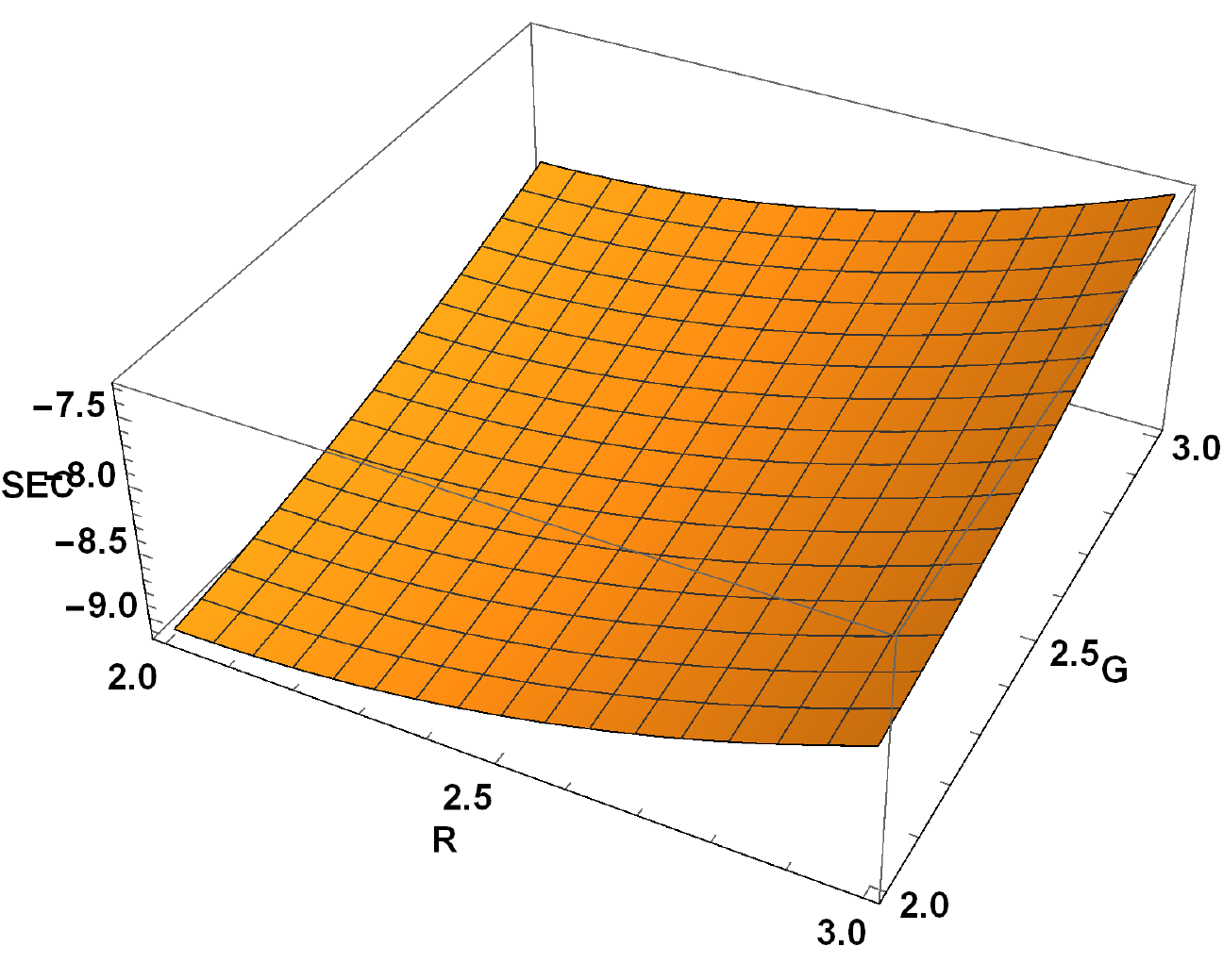}}
\caption{Energy conditions for $f(R,\mathcal{G})=R^2+\mathcal{G}ln\mathcal{G}$ with $\lambda=1$, $2\leq R \leq 3$ and $2\leq \mathcal{G} \leq 3 $}
\label{fig2}
\end{figure}
\newpage
From Figure \ref{fig2} it is clear that the behavior of density parameter is positive, DEC also shows the positive behavior while the SEC violates it. In this case, the equation of state parameter turns out to be $\omega=-1$.

\section{\textbf{Discussion}}

Nowadays, to describe the observed late-time cosmic acceleration the modified theories of gravity are considered as the most valid candidate. In this article, we have studied the effect of the geometrical constraint of $(GR)_4$ type spacetime on an interesting gravity theory, namely, the $f(R,\mathcal{G})$-gravity. To examine the self consistencies of these modified theories of gravity, the different energy conditions are favorable candidates and these energy conditions can be obtained via the popular Raychaudhuri equation. In this manuscript, we have investigated the null, dominant, weak, and strong energy conditions under $f(R,\mathcal{G})$-gravity in a conformally flat generalized Ricci recurrent perfect fluid space-time. We have taken two different $f(R,\mathcal{G})$ models, $f(R,\mathcal{G})= R+\sqrt{\mathcal{G}}$ and another one is $f(R,\mathcal{G})= R^2+\mathcal{G}ln\mathcal{G}$. We conclude that NEC is always satisfied in the present setting, weak and dominant energy conditions are satisfied when $-5\leq R \leq 1.5$ , $1\leq \mathcal{G} \leq 3.5 $ and $2\leq R \leq 3$ , $2\leq \mathcal{G} \leq 3 $ for the two models respectively provided with $\lambda>0$ while the strong energy condition shows the negative behavior for both the models. Hence both the models have good agreement with the current observational findings which depict the accelerated expansion of the universe.


\textbf{Acknowledgement:} A.D. and L.T.H. are supported by the grant \\  FRGS/1/2019/STG06/UM/02/6. R.S. acknowledges University Grants Commission (UGC), Govt. of India, New Delhi, for awarding Junior Research Fellowship (NFOBC) (No.F.44-1/2018(SA-III)) for financial support. The authors are grateful to the referee and editor  for their valuable suggestions towards the improvement of the paper.

\end{document}